\documentclass{birkjour_t2}

\usepackage{amsfonts}
\usepackage{amsmath}

\newtheorem{lemma}{Lemma}
\newtheorem{definition}{Definition}

\newcommand{\argmax}{\operatornamewithlimits{argmax}}

\begin{document}

\title[Faster Average Case Low Memory Semi-External Con\-struc\-tion of the BWT]{Faster Average Case Low Memory Semi-External\\Construction of the Burrows-Wheeler Transform}

\author{German Tischler}

\address{The Wellcome Trust Sanger Institute\\
Wellcome Trust Genome Campus\\
Hinxton, Cambridge\\
CB10 1SA\\
United Kingdom
}

\email{german.tischler@sanger.ac.uk}

\thanks{%
Supported by the Wellcome Trust.\\
Full version of an extended abstract which appeared
in the Proceedings of the 2\textsuperscript{nd} {I}nternational {C}onference on {A}lgorithms for {B}ig {D}ata.
}

\begin{abstract}
The Burrows Wheeler transform has applications in data compression as well as full text indexing. 
Despite its important applications and various existing algorithmic approaches the construction 
of the transform for large data sets is still challenging. In this paper we present a new semi external
memory algorithm for constructing the Burrows Wheeler transform. 
It is capable of constructing the transform for an input text of length $n$
over a finite alphabet in time $O(n\log^2\log n)$ on average, if sufficient internal memory is available to hold a
fixed fraction of the input text. In the worst case the run-time is $O(n\log n \log\log n)$.
The amount of space used by the algorithm in external memory is $O(n)$ bits.
Based on the serial version we also present a shared memory parallel algorithm running in time
$O(\frac{n}{p}\max\{\log^2\log n+\log p\})$ on average when $p$ processors are available.
\end{abstract}

\maketitle

\section{Introduction}
The Burrows-Wheeler transform (BWT) was introduced to facilitate the lossless
compression of data (cf.~\cite{burrows1994block}). It has an intrinsic
connection to some data structures used for full text indexing like the
suffix array (cf.~\cite{manber1993suffix}) and is at the heart of some compressed
full text self indexes like the FM index (see \cite{ferragina2000opportunistic}).
The FM index requires no more space than the k'th order entropy compressed
input text plus some asymptotically negligible supporting data structures.
Many construction algorithms for the BWT are based on
its relation to the suffix array, which can be computed from the input text
in time linear in the length of that text (see e.g.~\cite{karkkainen2003simple,nong2011two}).
While these algorithms run in linear time and are thus theoretically optimal
they require $O(n\log n)$\footnote{by $\log$ we mean $\log_2$ in this paper} bits of space for the uncompressed suffix array
given an input text of length $n$ while the text itself can be stored in a space
of $n\lceil\log\sigma\rceil$ bits for an alphabet of size $\sigma$ where we often have $\sigma\ll n$ and in
most applications $\sigma$ is constant. Algorithms for computing the suffix
array in external memory have been proposed (see e.g.~\cite{DBLP:conf/alenex/BingmannFO13,DBLP:journals/jea/DementievKMS08})
but these algorithms require large amounts of space and input/output in external memory.
An asymptotically optimal internal memory solution concerning time and space has been proposed \cite{DBLP:conf/focs/HonSS03}.
However the space usage of this algorithm is $O(n)$ bits for constant
alphabets, where an inspection of the algorithm suggests that the actual
practical memory usage of the algorithm is several times the size of the text in bits.
The practical space usage of the algorithm subsequently presented in \cite{DBLP:conf/spire/OkanoharaS09} is lower
(i.e. the involved constants are smaller) while theoretically not linear. It
however still requires multiple times as much space as the input text. A
sample implementation given by the authors only works for input sizes of up to
$2^{32}$ (see \cite{DBLP:conf/spire/BellerZGO13}) and only handles a single
level of the recursive algorithm. Given the implementation complexity of the algorithm it
remains unclear if it would scale well. 
Crochemore et al present an algorithm computing the BWT in quadratic
time with constant additional space (see \cite{DBLP:conf/cpm/CrochemoreGKL13}).
In \cite{DBLP:conf/spire/BellerZGO13}
Beller et al propose a semi external algorithm for the construction of the BWT based on induced
sorting. An algorithm is called semi external if it uses non negligible
amounts of internal as well as external memory. 
According to the authors the algorithm scales to arbitrary input sizes and
uses about one byte (i.e.~$8$ bits) per input symbol in internal memory. 
An algorithm constructing the BWT of a text by
block-wise merging using a finite amount of internal memory is presented in
\cite{DBLP:journals/algorithmica/FerraginaGM12}. The algorithm partitions
the text into a set of fixed size blocks. The run-time is $O(n^2/b)$
for a block size of $b$ and a text length of $n$.
It requires an amount of internal memory which is roughly sufficient to compute 
the suffix array of a single of these block.
In particular the amount of internal memory used can be smaller than the
space required for the text. In this paper we modify this algorithm to run
in time $O(n \log n \log\log n)$ in the worst case and $O(n\log^2\log n)$
on average for the case where we are able to keep a fixed fraction of the text in memory.
Assuming the practically common case of a finite alphabet the algorithm in
\cite{DBLP:journals/algorithmica/FerraginaGM12} uses blocks of size $O(n/\log n)$ when
provided with $O(n)$ bits of space in internal memory so its total run-time
for this setting is $O(n\log n)$. In consequence our algorithm is faster
on average and slower by $O(\log\log n)$ for a very unlikely worst case.
Compared to the algorithm presented in \cite{DBLP:conf/spire/BellerZGO13} our 
algorithm can work with less internal memory. 
For DNA for instance the complete text can be held in memory using about $2$
bits per symbol which is significantly less than a full byte (8 bits) per character.
We also propose a parallel version of our algorithm for the shared memory
model. This algorithm runs in time $O(\frac{n}{p}\max\{\log^2 n,\log p\})$ on
average and with high probability when using $p$ processors.

The rest of this paper is structured as follows. In Section \ref{sect:definitions} we introduce
definitions which we will be using throughout the paper. Section \ref{bwtmergingsection}
gives a short overview over the algorithm of Ferragina et al (cf.~\cite{DBLP:journals/algorithmica/FerraginaGM12})
on which our new algorithm is based. In Section \ref{computingperiods} we
present an algorithm for computing the minimal periods of all prefixes of a
string in succinct space. We show how to compute the suffix order for
blocks of $b$ suffixes drawn from a circular version of a string of length $n\geq b$ in optimal time $O(b)$ and space $O(b\log b)$ bits
after collecting some repetition information about the string in a preprocessing phase
taking time $O(n)$ and $O(b\log b)$ bits of space in Section \ref{singleblocksorting}.
In Section \ref{mergepairsect} we show how to merge the BWTs of two adjacent text blocks of
length $b_l$ and $b_r$ respectively given by their BWT in external memory
in time $O\left((b_l+b_r)\log\log(b_l+b_r)\right)$ and sufficient space in
internal memory to hold a rank dictionary for the left block. We describe
in Section \ref{bwtbalanced} how to combine the block sorting from Section
\ref{singleblocksorting} and block merging from Section \ref{mergepairsect}
to obtain a BWT construction algorithm based on a balanced instead of a skewed merge tree.
In Section \ref{sect:parallelisation} we describe a parallel version of our algorithm.
Finally in section \ref{conclusion} we summarise our results and discuss
some remaining open problems.
\section{Definitions}
\label{sect:definitions}
For a string $s=s_0s_1s_2\ldots s_{m-1}$ of length $|s|=m\in\mathbb{N}$ we define $s[i]=s_i$ for $0\leq i < m$ and for
$s=s_0s_1\ldots $ we define $s[i]=s_i$ for $0\leq i$.
For a finite word $u$ and a finite or infinite word $v$ we write their concatenation as $uv$.
For any finite words $u,x$ and finite or infinite words $w,v$ such that $w=uxv$ we call $u$ a
prefix, $v$ a suffix and $x$ a factor of $w$.
A prefix, suffix or factor of some string $w$ is called proper if it is not
identical with $w$.
The empty word consisting of no symbols is denoted by $\epsilon$.
For a string $s$ and indices $0\leq i \leq j < |s|$ we denote the factor $s[i]s[i+1]\ldots s[j]$ by
$s[i,j]$. For any $i,j$ such that $i>j$ the term $s[i,j]$ denotes the empty word.
A finite word $w$ has period $p$ iff $w[i]=w[i+p]$ for $i=0,\ldots,|w|-p-1$
and an infinite word $w$ has period $p$ iff $w[i]=w[i+p]$ for $i=0,1,\ldots$.
For a finite word $u$ and $k\in\mathbb{N}$ the $k$'th power
$u^k$ of $u$ is defined by $u^0=\epsilon$  and $u^{i+1}=u^i u$ for $i\in\mathbb{N}$.
A word $w$ is primitive if it is not a power of a word $u$ such that $|u|<|w|$.
A word $u$ is a root of $w$ if $w=u^k$ for some $k\in\mathbb{N}$.
A word $w$ is a square if there is a word $u$ such that $w=u^2$.
Throughout this paper let $\Sigma=\{0,1,\ldots,\sigma-1\}$ denote a finite alphabet for some $\sigma\in\mathbb{N},\sigma>0$
and let $t=t_0t_1\ldots t_{n-1}\in\Sigma^n$ denote a finite string of length $n>0$.
We define the semi infinite string $\tilde{t}$ by $\tilde{t}[i] = t[i - \lfloor i/n\rfloor n]$ for $i \geq 0$.
We define the suffix $\tilde{t}_i$ of $\tilde{t}$ as $\tilde{t}_i=\tilde{t}[i]\tilde{t}[i+1]\ldots$
and $\tilde{t}_i < \tilde{t}_j$ for $i,j\in\mathbb{N},i\ne j$ iff
either $\tilde{t}_i = \tilde{t}_j$ and $i < j$ or for the smallest $\ell\geq 0$ 
such that $\tilde{t}[i+\ell]\ne \tilde{t}[j+\ell]$ we have $\tilde{t}[i+\ell] <
\tilde{t}[j+\ell]$. The suffix array $A$ of $t$ is defined as the
permutation of the numbers $0, 1, \ldots, n-1$ such that 
$\tilde{t}_{A[i]} < \tilde{t}_{A[i+1]}$ for $i=0,1,\ldots,n-2$ and the
Burrows-Wheeler transform (BWT) $B=b_0b_1\ldots b_{n-1}$ of $t$ is given by $b_i = \tilde{t}[A[i]+n-1]$
for $i=0,1,\ldots,n-1$.
\section{BWT construction by block-wise merging}
\label{bwtmergingsection}
We give a short high level description of the algorithm by Ferragina et al.~in
\cite{DBLP:journals/algorithmica/FerraginaGM12} as we will be modifying it.
Unlike our algorithm it assumes the input string to have a unique minimal terminator symbol.
Given a block size $b$ the input string $t$ is partitioned
into $c=\lceil n/b \rceil$ blocks $T_0,T_1,\ldots,T_{c-1}$ of roughly equal size. 
The algorithm starts by suffix sorting the last block, computing its BWT $B_{c-1}$ and the bit array $gt_{c-1}$
which denotes for each suffix in $T_{c-1}$ but the first whether it is
smaller or larger than the first. 
The BWT of $T_i\ldots T_{c-1}$ for $i<c-1$ is computed
by first computing the suffix array for the suffixes starting in $T_i$ by
using the text of $T_i$ and $T_{i+1}$ in memory and handling the comparison
of suffixes starting in $T_i$ but equal until both have entered $T_{i+1}$
by using the bit vector $gt_{i+1}$ which explicitly stores the result of
this comparison.
The BWTs of $T_i$ and $T_{i+1}T_{i+2}\ldots T_{c-1}$
are merged by computing the ranks of the suffixes starting in $T_{i+1}T_{i+2}\ldots T_{c-1}$
in the sorted set of suffixes of $T_i$ and computing a gap array $G_i$ which
denotes how many suffixes from $T_{i+1}T_{i+2}\ldots T_{c-1}$ are to be
placed before the suffixes in $T_i$, between two adjacent suffixes in $T_i$
and after all suffixes in $T_i$.
This process follows a backward search of $T_{i+1}T_{i+2}\ldots T_{c-1}$ in $T_i$. 
Using the array $G_i$ it is simple to merge the two BWTs together.
For computing the rank of a suffix from $T_{i+1}\ldots T_{c-1}$ it is necessary to
know whether it is smaller or larger than the one at the start of
$T_{i+1}T_{i+2}\ldots T_{c-1}$ as $B_i$ is not a conventional BWT.
For further details about the algorithm the reader is referred to \cite{DBLP:journals/algorithmica/FerraginaGM12}.
\section{Computing the Minimal Period of the Prefixes of a String}
\label{computingperiods}
A border of a string $w$ is a string $u$ such that $w=uv=xu$ and $|v|=|x|\ne 0$.
The border array $\mathcal{B}$ of a string $w$ is the integer array of length $|w|$ such that
$\mathcal{B}[i]$ stores the length of the longest border of $w[0,i]$ for $i=0,1,\ldots,|w|-1$.
When a string $w$ has the border array $\mathcal{B}$, then the minimal
period of $w[0,i]$ for $0\leq i < |w|$ can be obtained as $i+1-\mathcal{B}[i]$
(see e.g. Proposition 1.5 in \cite{crochemore2007algorithms}).
The uncompressed border array for a string of length $b$ can be stored using $b$ words, i.e.~$b\lceil\log b\rceil$ bits.
Due to the properties of the border array (the value at index $i+1$ can be at
most $1$ larger than the value at index $i$ and all values are non-negative) it
is however easy to use a succinct version in $O(b)$ bits. 
For the rest of the section let $w$ denote a string and let $\mathcal{B}$ denote the border array of $w$.
Then we can represent $\mathcal{B}$ as the bit vector
$\tilde{\mathcal{B}}=1 0^{\mathcal{B}[0]-\mathcal{B}[1]+1} 1 0^{\mathcal{B}[1]-\mathcal{B}[2]+1} 1 \ldots 0^{\mathcal{B}[b-1]-\mathcal{B}[b-2]+1} 1$
and $\mathcal{B}[i]$ can be retrieved as $i-\textnormal{rank}_0(\tilde{\mathcal{B}},\textnormal{select}_1(\tilde{\mathcal{B}},i))$
for $i=0,1,\ldots,b-1$ where $\textnormal{rank}_0(\tilde{\mathcal{B}},i)$
denotes the number of $0$ bits in $\tilde{\mathcal{B}}$ up to and including
index $i$ and $\textnormal{select}_1(\tilde{\mathcal{B}},i)$ returns the
index of the $i+1$'th $1$ bit in $\tilde{\mathcal{B}}$. 
Indexes for the rank and select operations can be stored in $o(b)$ extra bits (see e.g.~\cite{navarro2007compressed}).
The most common indexes for these operations are tailored for static bit
vectors while the computation of the border array (see for instance section 1.6 in \cite{crochemore2007algorithms})
incrementally extends the array while using the already existing prefix to deduce the next value.
The rank and select indexes as presented in section 6.1 in
\cite{navarro2007compressed} are however easily adapted for the case of only appending bits at the end of the bit vector.
Both types of indexes partition the indexed bit vector into blocks. In the
case of rank the blocks consist of $\beta_0 = \lceil{\log^2 |\tilde{\mathcal{B}}|}\rceil$ bits of $\tilde{\mathcal{B}}$.
An array $\mathcal{R}_0$ is used to explicitly store $\textnormal{rank}_0(\tilde{\mathcal{B}},\beta_0 i)$ for $i=0,1,\ldots,\lceil\frac{n}{\beta_0}\rceil-1$
where each element of $\mathcal{R}_0$ takes $\lceil\log |\tilde{\mathcal{B}}| \rceil$ bits.
The blocks are again partitioned into smaller blocks of length $\beta_1=\lceil\frac{\log |\tilde{\mathcal{B}}|}{2}\rceil$
(for the sake of simplicity in the discussion we assume $\beta_1$ divides
$\beta_0$. The method is easily extended to the case where $\beta_0$ is not exactly a multiple of $\beta_1$).
Another integer array $\mathcal{R}_1$ is used to store
$\textnormal{rank}_0(\tilde{\mathcal{B}},\beta_1 i) - \textnormal{rank}_0(\tilde{\mathcal{B}},\lfloor\frac{\beta_1 i}{\beta_0}\rfloor)$
for $i=0,1,\ldots,\lceil\frac{n}{\beta_1}\rceil - 1$. Each element of $\mathcal{R}_1$ takes $\lceil \log \beta_0 \rceil$ bits.
We can incrementally build the rank dictionary by adding one value to $\mathcal{R}_0$ and
$\lceil\frac{\beta_0}{\beta_1}\rceil$ values to $\mathcal{R}_1$
each time $\beta_0$ bits have been appended to the underlying bit vector.
This can trivially be done in time $O(\beta_0)$ for each block of size $\beta_0$.
The answers for a block yet to be completed can be stored explicitly in an
array of size $\beta_0$ which takes $\beta_0 \lceil\log\beta_0\rceil = O(\log^2
\tilde{\mathcal{B}} \log \tilde{\mathcal{B}} ) $ bits and thus is
asymptotically negligible compared to the space used for the bit vector itself.
Using this method the rank operation can be computed in time $O(1)$ at any
time and the time for adding each single bit to the vector is amortised constant.
The method for keeping constant time select operations is conceptually very
similar. The data structure partitions the bit vector $\tilde{\mathcal{B}}$
such that each block contains $\zeta_0=\lceil\log^2 |\tilde{\mathcal{B}}|\rceil$
bits of value $1$. Again we can keep enough space to store the answers for a
single incomplete block explicitly and extend the index for select each time
the total number of $1$ bits in the intermediate bit vector reaches an
integer multiple of $\zeta_0$.
\begin{lemma}
Let $x$ denote a string over a finite alphabet $\Gamma=\{0,1,\ldots,\gamma-1\}$.
The border array of $x$ can be computed in time $O(|x|)$ and using space
$O(|x|\log\gamma)$ bits.
\end{lemma}
\begin{lemma}
\label{periodenumlemma}
Let $x$ denote a string over a finite alphabet $\Gamma=\{0,1,\ldots,\gamma-1\}$.
The sequence of minimal periods of the prefixes of $x$ can be enumerated in
time $O(|x|)$ and space $O(|x|\log\gamma)$ bits.
\end{lemma}
Using the periodicity lemma (see Lemma 1 in \cite{knuth1977fast}) we in
addition obtain the following lemma.
\begin{lemma}
\label{uniqueperiod}
Let $x$ denote a string over a finite alphabet $\Gamma=\{0,1,\ldots,\gamma-1\}$.
If $x$ has any period $q \leq \lfloor\frac{|x|}{2}\rfloor$ then there exists
a minimal period $p$ of $x$ such that $p\leq \lfloor\frac{|x|}{2}\rfloor$
and $p$ divides all other periods of $x$ whose value does not exceed $\lfloor\frac{|x|}{2}\rfloor$.
It is decidable in time $O(|x|)$ and space $O(x\log\gamma)$ bits whether $x$
has any period $p\leq \lfloor\frac{|x|}{2}\rfloor$ and if it has any such
periods then their minimum can be computed in the same time and space bounds.
\end{lemma}
\begin{proof}
Let $x$ denote a string over a finite alphabet $\Gamma=\{0,1,\ldots,\gamma-1\}$
and let $P$ denote the set of all periods of $x$ whose value does not exceed
$\lfloor\frac{|x|}{2}\rfloor$. If $P$ is empty, then the lemma holds
as $x$ has no relevant periods. Otherwise let $p$ be the minimal element of $P$.
Then for any element $q$ of $P$ we have
$p+q \leq |x|$ as $p \leq q \leq \lfloor\frac{|x|}{2}\rfloor$. According to the
periodicity lemma this implies that the greatest common divisor of $p$ and
$q$ is also a period of $x$. As $P$ contains all the periods of $x$ not
exceeding $\lfloor\frac{|x|}{2}\rfloor$, $p$ is the minimal element of $P$
and the greatest common divisor of $p$ and $q$ divides $p$ it follows that
this greatest common divisor is $p$. As $p$ is the minimal period of $x$ it
can be computed in time $O(|x|)$ and space $O(|x|\log \gamma)$ bits according to
Lemma \ref{periodenumlemma}.
\end{proof}

\section{Sorting single blocks}
\label{singleblocksorting}
The algorithm by Ferragina et al processes each single block relying on
knowledge about the priorly fully processed following block, in case of the
last block the terminator.
For our algorithm we need to be able to sort a single block without knowing the complete sorted order of the next block.
For this purpose we need to be able to handle repetitions, one of the major
challenges along the way, efficiently.
For the following argumentation we will need the occurring block sizes to be
as similar as possible, so we will deduce a final block size from a target
block size chosen to satisfy memory constraints.
Let $b'$ denote a preliminary target block size such that $0<b'\leq n$. 
From $b'$ we can deduce a final block size $b$ as 
\begin{equation}
	b=\left\lceil\frac{n}{\left\lceil\frac{n}{b'}\right\rceil}\right\rceil\enspace
.
\label{blocksizeeq}
\end{equation}
In Equation \ref{blocksizeeq} the denominator denotes the number $\nu$ of blocks
required for the preliminary block size $b'$. The final block size is
obtained as the smallest number $b$ of symbols required to keep this number of blocks
$\nu$, in particular we have $b \leq b'$
As $b$ is the smallest such number we have $\mu' = \nu b - n$ such that $0 \leq \mu' < \nu$
and equally $\mu = \nu - \mu'$ such that $1 \leq \mu \leq \nu$.
In consequence we can decompose the set of indices $[0,n)$ on $t$ into
$\mu$ sub intervals of length $b$ and $\nu-\mu$ sub intervals of length $b-1$
and define the index blocks $\mathcal{B}_i$ by
\begin{equation}
	\mathcal{B}_i = \left\lbrace
\begin{array}{ll}
\lbrack ib, (i+1)b )                                  & \textnormal{for }  0 \leq i < \mu \\
\lbrack \mu b + (i-\mu) (b-1),\mu b + (i+1-\mu) (b-1) & \textnormal{for }\mu \leq i < \nu
\end{array}
\right.
\end{equation}
Informally we have $\mu$ blocks of size $b$ and $\nu-\mu$ blocks of size 
$b-1$ such that the blocks of size $b$ precede the ones of size $b-1$.
Blocks of size $b-1$ only exist if $b$ does not divide $n$.
When comparing two suffixes of $\tilde{t}$ starting at indices $i$ and $j$ such that $i < j$
we may encounter two cases. 
In the first case we have $\tilde{t}[i,j-1] \ne \tilde{t}[j,j+(j-i-1)]$.
Then the comparison requires handling at most $j-i$ symbol pairs, i.e.~at
most $2b-2$ symbols if the two suffixes start in the same block of length $b$ as defined above.
In the second case we have $\tilde{t}[i,j-1] = \tilde{t}[j,j+(j-i-1)]$, i.e.~a
square with period $j-i$ at index $i$ in $\tilde{t}$.
Thus situations which require us to have access to $\omega(b)$ symbols for sorting the suffixes of
a block of indices on $\tilde{t}$ are induced by repetitions of periods
strictly smaller than $b$.
In consequence only repetitions with a period smaller than $b$ are relevant for our block sorting.
\begin{definition}
Let $\mathcal{B}$ denote a block of $b$ indices on $\tilde{t}$ starting at index $i$, 
	i.e.~the set of suffixes $\tilde{t}_{i+j}$ for $j=0,1,\ldots,b-1$.
\begin{itemize}
\item $\mathcal{B}$ propagates a repetition of period $p, 1\leq p \leq b$ iff $\tilde{t}_i[0,b+2p-1]$ has period $p$.
\item $\mathcal{B}$ generates a repetition of period $p, 1 \leq p \leq b$ iff $\tilde{t}[b-p,b-1] = \tilde{t}[b,b+p-1]$
      and the block of $b$ suffixes starting at $i+b$ propagates a repetition of period $p$.
\end{itemize}
\end{definition}
If a block propagates repetitions of any periods, then there is a unique minimal period
dividing all other propagated periods according to Lemma \ref{uniqueperiod}.
This unique minimal period can then be computed in time $O(b)$ and space $O(b\log\sigma)$ 
bits using minor modifications of standard string algorithms (see Section \ref{computingperiods}).
As there is a unique minimal period propagated by a block if any and 
as for repetition generation we are
only interested in periods which are propagated by the next block
we can compute the relevant generation properties of a block in the same
time and space bounds.
The suffix order on $\tilde{t}$ defined as above defines the order of two
suffixes at indices $i$ and $j$ such that $i\ne j$ by index comparison if $\tilde{t}_i = \tilde{t}_j$.
This case can only appear if the text $t$ itself has a period $p$ which is
smaller than $n$ and divides $n$, i.e.~when $t=\alpha^i$ for some string $\alpha$ and some integer $i>1$.
For sorting the suffixes in our blocks we can avoid this case. If $t$ is an
integer power of a string $\alpha$ such that $|\alpha| < b$ for some exponent $k>1$, then we will
observe that each of our blocks propagates period $p$, thus this case is
easily detected given precomputed repetition propagation information. In
this case the Burrows-Wheeler transform of $t$ can be obtained from the
Burrows-Wheeler transform for $\alpha$ by repeating each symbol $k$ times,
e.g. if $abc$ is the Burrows-Wheeler transform for $\alpha$ and we have
$k=2$ then the transform for $t$ is $aabbcc$.  Similar arguments can be employed to
obtain the suffix array and inverse suffix array of $t$ from the reduced
variants for $\alpha$.
Thus in the following we will without loss of generality assume that $t$ is not an integer power of a
string of period smaller than $b$. In consequence all suffix comparisons
required for sorting the suffixes starting in a block are decided by
character inequality and not by position.
Using information about short period repetitions in the input string, we are
able to handle the sorting of a single block of suffixes extending beyond
the end of the block efficiently by reducing long repetitions as we show in
the following lemma.
\begin{lemma}
A block of $b$ circular suffixes of $\tilde{t}$ can be sorted in
lexicographical order using time $O(b)$ and space $O(b\log b)$ bits
using precomputed repetition propagation data.
\label{suffixsortinglemma}
\end{lemma}
\begin{proof}
The pairwise order between all suffixes of the block is decided by character
inequalities. However these character inequalities may happen after $\Omega(n)$ symbols in each case.
Note however that if we compare two strings $u$ and $v$ lexicographically such that
$u=\alpha \beta^{k} \gamma$ and $v=\alpha \beta^{k}\delta$ we obtain the
same result as when comparing $u'=\alpha \gamma$ and $v'=\alpha\delta$, i.e.
as when we remove the repetition of $\beta$ from both.
Now assume we are sorting the circular suffixes of $t$ for a block starting
at index $i$ with block length $b$, i.e.~we are sorting $b$ strings of
infinite length where for each pair of strings the comparison ends after
less than $n$ steps.
When comparing a pair of such strings starting at offsets $j$ and $k$  in the
block such that without loss of generality $j<k$ we may encounter the following cases.
\begin{enumerate}
\item There exists $\ell < (b-j)+(k-j)=b+k-2j$ such that $\tilde{t}[i+j+\ell]\ne \tilde{t}[i+k+\ell]$.
Then as $b+k-2j < 2b$ the order between the strings would be correctly
reflected in a suffix sorting of $\tilde{t}[i,i+2b-1]$, where $2b$ is $O(b)$.
\item Otherwise we have 
$\tilde{t}[i+j,i+k-1] = \tilde{t}[i+k,i+k+(k-j)-1]$,
i.e.~a square of period $k-j$ at position $i+j$ in $\tilde{t}$ and this
repetition extends to at least index $i+b+(k+j)-1$ in $\tilde{t}$.
There is thus a complete instance of a rotation of the root of the square in $\tilde{t}$ 
after the end of the considered block (the root of the repetition may
overlap the block boundary or come to lie outside of the block).
In consequence the block may generate a repetition of period $k-j$ depending
on whether the following block propagates a repetition of this period. If
$\tilde{t}[i+j,i+k-1]$ is an integer power
of a shorter primitive string $u$, then let $p=|u|$, otherwise let $p=k-j$.
Now we may encounter
two sub cases. If the following block does not propagate period $p$, then
the order between our two strings starting at index $i+j$ and $i+k$ in
$\tilde{t}$ would be correctly reflected in a suffix sorting of
$\tilde{t}[i,i+2b+2p-1]$ due to the definition
of the propagation of a repetition. Note that $2b+2p-1$ is $O(b)$ as $p \leq k-j < b$.
Now assume the following block does propagate period $p$. Let
$\beta=\tilde{t}[i+b,i+b+p-1]$ and let $\alpha=\tilde{t}[i,i+(b-p)-1]$. Then the
block can be written as $\alpha\beta$. Further we know that
$\tilde{t}[i]\tilde{t}[i+1]\ldots = \alpha \beta^{m} \gamma$ for some integer $m>1+\lceil b/p \rceil$ and $\gamma[0,p-1]\ne \beta$.
Using precomputed information about the single blocks' repetition propagation
properties and adequate data structures it is simple to deduce the numbers
$m$ and $\gamma$ in time $O(b)$. For each block propagating a repetition we
store the index of the closest following block in $\tilde{t}$ which no
longer propagates this repetition in an array. This array can be computed in
time $O(n)$ with a single scan over the list of repetitions propagated by
the single blocks. Even for a moderate block size in $O(\log^2 n)$ the space
required for this array is asymptotically negligible compared to the text
and it can be stored in external memory. 
Now for computing $m$ and $\gamma$  we can go to the last following block still propagating 
the repetition and perform a naive scan of
the text until the repetition ends. This scan will terminate in time $O(b)$.
For determining the order of the suffixes in the current block it is
sufficient to sort $\alpha\beta^{1+\lceil b/p \rceil}\gamma$ which has
length $O(b)$.
\end{enumerate}

\end{proof}
Sorting the suffixes of $\tilde{t}$ starting in a given block of length $b$
using the precomputed repetition generation and propagation block properties
can thus be performed in time $O(b)$ using space $O(b\log b)$ bits in internal memory.
Given the explicit suffix sorting for a block,
it is trivial to determine, whether any other suffix in the block is
lexicographically smaller or larger than the first suffix of the block 
and store the resulting bit vector in external memory.
For forward searching using the suffix array it is useful to in addition have the
longest common prefix (LCP) array.
For two finite strings $u,v$ let $\textnormal{LCP}(u,v)=\argmax_{l=0}^{\min\{|u|,|v|\}} u[0,\ell-1] = v[0,\ell-1]$.
For two infinite strings $u,v$ let $\textnormal{LCP}(u,v)=\infty$ if $u=v$
and $\textnormal{LCP}(u,v)=i$ where $i$ is the smallest non negative integer such that $u[i]\ne v[i]$ otherwise.
For a block of indices $[i,i+b-1]$ on $\tilde{t}$ such that $i\in\mathbb{N}$ let $\mathcal{A}$
denote the permutation of $i,i+1,\ldots,i+b-1$ such that $\tilde{t}_{\mathcal{A}[j]} < \tilde{t}_{\mathcal{A}[j+1]}$ for $j=0,1,\ldots,b-2$.
Then the LCP array of the block is defined by $\textnormal{LCP}[0]=0$ and
$\textnormal{LCP}[i]=\textnormal{LCP}(\tilde{t}_{\mathcal{A}[i-1]},\tilde{t}_{\mathcal{A}[i]})$
for $i=1,2,\ldots,b-1$.
Using a repetition reduction method similar to the suffix sorting case 
we obtain the following result.
\begin{lemma}
\label{computinglongestcommonprefixes}
The LCP array for a block of $b$ circular suffixes on $\tilde{t}$ can be
computed in time $O(b)$ and space $O(b\log b)$ bits
using precomputed repetition propagation data.
\end{lemma}
\begin{proof}
As the input string is without loss of generality not an integer power of a
string of period $p < b$ all suffix comparisons on $\tilde{t}$ within a
block are decided by character inequalities.
In consequence the values in the LCP array of a block are finite but the array can contain values in $\Omega(n)$
as the considered suffixes extend beyond the end of the block.
Given a finite string $s$ and the suffix array of $s$ the LCP array for $s$
can be computed in time $O(|s|)$ and stored in an array of size $O(s\log s)$ bits
(see \cite{DBLP:conf/cpm/KasaiLAAP01}).
Thus if the block does not generate a repetition, then we can compute the LCP array for the
block in time $O(b)$ and represent it as an array using $O(b\log b)$ bits.
This involves extending the text block by $O(b)$ elements (see above), 
computing the LCP array using the suffix array and the text and filtering the LCP array
so that only suffixes inside the original block remain.
If we find suffixes outside the original block between two original suffixes than we need to
combine the minimum of the corresponding entries with the second original
block value, e.g.~if the block starts at index $i$ and we have 
$\mathcal{A}[j] < i+b, \mathcal{A}[j+1]\geq i+b, \mathcal{A}[j+2]\geq i+b, \ldots, \mathcal{A}[j+k] < i+b$ 
then the LCP value between the suffixes at the block indices $\mathcal{A}[j]$ and $\mathcal{A}[j+k]$
is $\min\{ \textnormal{LCP}[j+1], \textnormal{LCP}[j+2], \ldots, \textnormal{LCP}[j+k] \}$
and we drop the LCP values relating to suffixes outside the block.
Now assume the block generates a repetition of period $p$.
Then as described in case $2$ of the proof of Lemma \ref{suffixsortinglemma} it is sufficient to suffix sort the string 
$\alpha\beta^{1+\lceil b/p \rceil}\gamma$ as defined above
to obtain the order of the suffixes in the current block. 
This block has length $O(b)$, so its LCP array can be computed in time $O(b)$ and
stored in space $O(b\log b)$ bits. 
We first remove the values corresponding
to suffixes outside the block as we did for the non propagating case. 
Then the LCP value at index $j>1$ is correct if $\textnormal{LCP}[j] <
\min\{i+b-\mathcal{A}[j-1],i+b-\mathcal{A}[j]\}$, i.e.~if the LCP between
the two corresponding suffixes is too short to reach the next block for at
least one of the two.
Otherwise both suffixes are equal at least until the comparison of the two has 
extended to the next block. 
Then the stored LCP value will be too small and thus incorrect.
We can however easily correct these values using the precomputed repetition propagation information. We
can compute in time $O(b)$ at which offset from the start of the next block
the propagated repetition first breaks.
Let $o$ denote this offset. Then we obtain the correct LCP value for index $j>0$ as 
$\textnormal{LCP}[j] = (i+b-\max\{\mathcal{A}[j-1],\mathcal{A}[j]\}) + o$.
If we store this value explicitly, then the LCP array will take space
$O(b\log n)$ bits.
As we can however detect and correct the incorrect LCP
values in constant time given the suffix array and number $o$ we can leave
the LCP array as is taking space $O(b\log b)$ and still have constant time access
to the correct values by computing them as required.
\end{proof}
Using the suffix and LCP array the time for a forward search of a pattern of
length $m$ in a block of size $b$ reduces from $O(m\log b)$ to $O(m+\log b)$ (see \cite{manber1993suffix}).
\section{Merging Pairs of Adjacent Blocks}
\label{mergepairsect}
In our modified algorithm we replace the completely skewed binary merge tree used
in \cite{DBLP:journals/algorithmica/FerraginaGM12} by a balanced binary merge tree.
Consequently we will need to be able to merge blocks with a block size in $\Omega(n)$.
For merging two adjacent blocks we need the following components:
\begin{enumerate}
\item The BWT of the left and right block. These can be compressed and in external memory as they will be 
      scanned sequentially.
\item An internal memory index of the left block suitable for backward search in $O(1)$ time per step.
      An FM type index using space $b_l H_k + o(n\log\sigma)$ bits can be used
      where $b_l$ is the length of the left block and $H_k$ denotes the k'th order entropy of the left block
      (see \cite{navarro2007compressed}).
\item The $gt$ bit vectors for the left and right block.
      Scanned sequentially and thus can be read streaming from external memory.
\item The number of circular suffixes in the left block smaller than the rightmost suffix of the right block.
      Used as the start point for the backward search.
\item The gap array $G$.
\end{enumerate}
The first three are equivalent to those used in \cite{DBLP:journals/algorithmica/FerraginaGM12}.
The rank of the rightmost suffix in the right block relative to the suffixes
of the left block can be obtained by employing forward search on one or more
text blocks.
If the left block is a single block which was produced by explicit suffix sorting using the
method of Section \ref{singleblocksorting}, then the rank can be obtained
using classical forward search in the suffix array while using the adjoined
LCP array.
This takes time $O(n + \log b)$ in the worst case  (on average this can be expected to be $O(\log n + \log b)$, see \cite{szpankowski1991height}).
If the left block was already obtained by merging $c$ blocks together, then the
desired rank can be obtained as the sum of the ranks of the suffix relative
to all single blocks composing the left block in time $O(c(n+\log b))$.
Assuming the blocks are merged together in a balanced binary merge tree
the total time used for forward searches is $O(\frac{n}{b}\log\frac{n}{b}n)$
in the worst case and $O(\frac{n}{b}\log\frac{n}{b}\log n)$ on average.
If we choose $b\in O(\frac{n}{\log n})$ then this becomes $O(n\log n\log\log n)$.
The memory required for the index of the left block in internal memory will be 
$b_l\lceil\log\sigma\rceil+o(b_l\log\sigma)$ bits for a left block size of $b_l$ assuming that the entropy compression is ineffective. 
This leaves us with the space required for the gap array. 
In the original algorithm this is a conventional array in internal memory taking space $O(b\log n)$ bits for a
left block size of $b$.
As we want to be able to merge blocks with size in $\Omega(n)$ this space requirement is too high.
Using Elias $\gamma$ code (cf.~\cite{elias1975universal}) we can store the gap array for merging a left
and right block of length $b_l$ and $b_r$ respectively in $O(b_l+b_r)$ bits
of space as we show in the following lemma.
\begin{lemma}
\label{gammagapfulllemma}
Let $G$ denote an array of length $\ell$ such that $G[i]\in\mathbb{N}$ for $0\leq i < \ell$
and $\sum_{i=0}^{l-1} G[i]=s$ for some $s\in\mathbb{N}$. Then the $\gamma$ code for $G$ takes $O(\ell + s)$ bits.
\end{lemma}
\begin{proof}
Elias $\gamma$ code stores the number $z>0$ as the bit sequence $0^{\lfloor\log z\rfloor}\textnormal{bin}(z)$
of length $1+2\lfloor\log z\rfloor$
where $\textnormal{bin}(z)$ denotes the binary representation of $z$ (most significant to least significant bit left to right).
The $\gamma$ code for $G$ thus requires
\[
\ell + 2 \sum_{i=0}^{\ell-1} \lfloor \log (G[i]+1) \rfloor \enspace
\]
bits. As $\gamma$ code is unable to represent the
number zero but zero is a valid entry in $G$ we add $1$ to each
element. We have
\[
\sum_{i=0}^{b_l-1} \lfloor \log (G[i]+1) \rfloor \leq \sum_{i=0}^{b_l-1} \log (G[i]+1) = \log \prod_{i=0}^{b_l-1}
(G[i]+1)\enspace .
\]
and the last product is maximal for $G[0]=G[1]=\ldots =G[\ell-1]\approx \frac{s}{\ell}$ 
(conceptually due to the fact that $x(1-x)$ has its maximum at $x=(1-x)=\frac{1}{2}$, i.e.~a product of
two positive real numbers of constant sum is maximal if both are chosen as half the
sum). 
Thus the maximum space used for the gap array in $\gamma$ code is bounded by
\[
\ell + 2 \sum_{i=0}^{\ell-1} \lfloor \log G[i] + 1 \rfloor 
\leq \ell \left( 1 + 2 \log \left(\frac{s}{\ell} + 1\right) \right)
\]
For $s \leq \ell$ the argument of the logarithm is between $1$ and $2$ thus the
space usage is bounded by $3\ell$. For $s > \ell$ let $s=2^r l$ for some $r\in\mathbb{R},r>0$.
Then we have
\[
\begin{array}{lcl}
\ell ( 1 + 2 \log ( \frac{s}{l}+1 ) ) & \leq & \ell ( 1 + 2 + 2 \log \frac{s}{l} ) \\
&\leq&\frac{s}{2^r}(3+2r) = s \frac{3+2r}{2^r} \leq s ( 3 + \frac{1}{\log_e 2} ) \leq 5 s
\end{array}
\]
which is in $O(s)$. Thus in both cases the space is in $O(\ell+s)$ bits.
\end{proof}
Elias $\gamma$ code however is not suitable for efficient updating as we would need it for
computing the gap array.
During the computation of a gap array each step of the backward search leads to an increment of exactly one position in the gap array.
In particular the sum over the elements of the array after $s$ steps is $s$.
The $\gamma$ code representation of the gap array does not allow efficient
updating of single elements.
We overcome this problem by producing partial sparse gap arrays, which we write to external memory.
We accumulate an amount $c$ of indices for incrementing in internal memory
before we write a sparse variant to external memory, which implies that the
sum over the sparse array written will be $c$. Whenever we have produced two
sparse gap arrays in external memory featuring the same element sum, then we merge
the two together into a sparse array of twice that sum. At the end of the
process we need to merge the remaining set of arrays into a single final
gap array.
Let $G$ denote a gap array of length $\ell$ such that $\sum_{i=0}^{\ell-1} G[i] = s$ for 
some $s > 0$ with $k>0$ non-zero values.
Let $i_0,i_1,\ldots,i_{k-1}$ denote the sequence of indices of non-zero values in $G$ in increasing order.
We store two $\gamma$ coded bit streams for representing $G$. In the first we store the 
sequence $i_0+1,i_1-i_0,i_2-i_1,\ldots,i_{k-1}-i_{k-2}$
and in the second we store the subsequence of non-zero values in $G$.
Note that the elements of both sequences are strictly positive integers, so
we do not need to add $1$ for the $\gamma$ code.
The maximum space usage for the first sequence is bounded by
$k(1+2\log\frac{\ell}{k})$ and for the second by
$k(1+2\log\frac{s}{k})$ which sums up to $2k(1+\log\frac{s\ell}{k^2}))$ 
(see the proof for Lemma \ref{gammagapfulllemma} for details on obtaining these bounds).
For fixed $s,\ell$ this function has a maximum for $k=\frac{\sqrt{2s\ell}}{e}$
with value $\frac{4\sqrt{2s\ell}}{e\log_e 2}$ where $e$ denotes the base of the natural
logarithm. 
As we by definition have $k \leq \min\{\ell,s\}$ this maximum can
only be observed for $s \geq \frac{2\ell}{e^2} \approx .2706 \ell$ (and equally
$\ell \geq \frac{2s}{e^2}$). 
For $s < \frac{2\ell}{e^2}$ we obtain the maximum space usage for $k=s$ and this space maximum 
is then bounded by $2s(1+\log\frac{\ell}{s})$.
While computing a complete gap array we may at an intermediate stage have no
more than two arrays with a sum of $c 2^i$ for $i=0,1,\ldots,\lfloor\log\frac{s}{c}\rfloor$ 
when the arrays of the smallest sum we produce have sum $c$ for a constant $c \geq 1$.
We will never have two arrays for more then two of the sizes while merging
and only a single one for all other sizes, but for the sake of simplicity in
obtaining a bound on the size of the sum of the arrays we will assume that
two for each size are possible.
For those arrays with a sum of at most $\frac{\ell}{4}$ we can bound the space used by
\[
\begin{array}{lcl}
2 \sum_{i=0}^{\lfloor\log \frac{l}{4c}\rfloor} 2(2^ic)(1+\log\frac{\ell}{2^ic}) & \leq &
4c \sum_{i=0}^{\lfloor\log \frac{l}{4c}\rfloor} 2^i ( 1 + \log \frac{\ell}{2^ic} ) \\
& \leq & 4c \sum_{i=0}^{\lfloor\log \frac{\ell}{4c}\rfloor} 2^i ( 1 + \log \frac{\ell}{2^i } ) \\
& \leq & 8c \sum_{i=0}^{\lfloor\log \frac{\ell}{4c}\rfloor} 2^i ( \log \frac{\ell}{2^i } ) \\
& \leq & 8c \sum_{i=0}^{\lfloor\log \frac{\ell}{4c}\rfloor} 2^i ( \log \ell - i  ) \\
& \leq & 8c \sum_{i=0}^{\lceil\log \ell \rceil} 2^i ( \lceil \log \ell \rceil - i  ) \\
& \leq & 2^{\lceil\log \ell \rceil+3}c \sum_{i=0}^{\lceil\log \ell \rceil} \frac{i}{2^i} \\
& \leq & 2^{\lceil\log \ell \rceil+4}c \\
& \leq & 32 c \ell
\end{array}
\]
Thus the space for these arrays is $O(\ell)$ bits as $c$ is a constant. 
For partial arrays with a sum of more than $\frac{\ell}{4}$ we can resort to dense arrays. We may have
a constant number of arrays with a sum between $\frac{\ell}{4}$ and $\ell$.
These take space $O(\ell)$ (see the proof for Lemma \ref{gammagapfulllemma}). Finally
we can have a logarithmic number of arrays with a sum larger then $\ell$.
Analogously to the arrays with a sum of at most $\frac{\ell}{4}$ we can deduce
that the space for these arrays is $O(s)$ bits. In total all partial gap
arrays together take space $O(s+\ell)$ bits.
Two sparse partial gap arrays can be merged in time linear in the number of
their respective non-zero elements which is bounded by the sum over the
elements.
Two of the dense partial gap arrays we use can also be merged in time linear in the sum
over their elements, as these arrays are only used if the sum over their
elements reaches a quarter of there length.
Merging two partial gap arrays is in both cases a streaming operation and requires
virtually no internal memory.
If the merging process starts with partial arrays for sum $1$, the final
array has length $\ell$ and the sum over its elements is $s$
then the total merging time can be bounded by $\sum_{i=0}^{\lceil\log s\rceil}
\min\{s,\ell\}$.
For merging two blocks of size $b_l$ and $b_r$ respectively this is
in $O(\min\{b_l,b_r\}\log\max\{b_l,b_r\})$ which becomes $O(b_l\log b_l)$ if
$b_l\in\Theta(b_r)$ as it holds for balanced binary merge trees.
If we accumulate $\frac{b_r}{\log^2 b_r}$ indices for incrementing in an array
in internal memory we can reduce the time used for merging without
asymptotically increasing the total memory footprint of the algorithm.
The space used for the array in internal memory is $O(\frac{b_r}{\log b_r})$
bits and thus asymptotically negligible compared with the (uncompressed) text.
When the internal memory array is full, then we first need to sort it so we can
write the partial gap array representation. 
This can be performed in time linear in the length of the array by using a
two phase radix sort with a number of buckets in $O(\sqrt{b_r})$, which require
additional internal memory of $O(\sqrt{b_r})$ words or $O(\sqrt{b_r}\log b_r)$ bits.
For both run-time and space usage the radix sort does not asymptotically
increase the resources required for the merging.
Starting with partial gap arrays of sum $\frac{b_r}{\log^2 b_r}$ implies the
total number of such arrays is $O(\log^2 b_r)$. 
Consequently the binary merge tree for obtaining the final gap array has
a depth in $O(\log \log^2 b_r)=O(\log\log b_r)$ and it can be processed
in time $O(b_r \log\log b_r)$.

The $gt$ array for the merged block can be composed by concatenating the
$gt$ array for the left block and an array storing the respective
information for the right block computed while performing the backward search
for filling the gap array. 
For this purpose we need to know the rank of the leftmost suffix in the left block.
This can either be computed using forward
search on the suffix arrays of the basic blocks or extracted from a
sampled inverse suffix array which can be computed along the way.
The sampled inverse suffix arrays of two blocks can just like the BWTs of the
two blocks be merged using the gap array. This is also an operation based on
stream accesses, so it can be done in external memory in time $O(b_l + b_r)$
when merging two blocks of size $b_l$ and $b_r$.
\section{BWT Computation by Balanced Tree Block Merging}
\label{bwtbalanced}
Using the building blocks described above we can now describe the complete
algorithm for computing the BWT of $t$ by merging basic blocks according to
a balanced binary tree.
\begin{enumerate}
\item Choose a target block size $b'\in O(\frac{n}{\log n})$ and deduce a
block size $b=\lceil\frac{n}{\lceil\frac{n}{b'}\rceil}\rceil$ such that the number of blocks $c$ satisfies 
$c=\lceil\frac{n}{b}\rceil=\lceil\frac{n}{b'}\rceil$ and
$n$ can be split into blocks of size $b$ and $b-1$ only. Split $t$ such that
the blocks of size $b$ appear before those of size $b'$. This step takes
constant space and time.
\item Compute which blocks in $t$ propagate repetitions of period at most
$b$ and for each block which is followed by a block propagating a repetition
whether it is generating this repetition. This takes time $O(n)$ in total and
space $O(b\log\sigma) = O(\frac{n\log\sigma}{\log n})\subseteq O(n)$ bits. 
The result data can be stored in external memory.
\item Compute a balanced merge tree for the blocks. Start with a root
representing all blocks. If a node containing a single block is considered
produce a leaf and stop. Otherwise for an inner node representing $k>1$ blocks
produce a left subtree from the $\lceil\frac{k}{2}\rceil$ leftmost blocks and a right
subtree from $\lfloor\frac{k}{2}\rfloor$ rightmost blocks in $t$. The
tree has $O(\log n)$ nodes. Each node stores at most two (start and end) block
indices taking $O(\log\log n)$ bits and two node pointers also taking space
$O(\log\log n)$ bits. So the total tree takes space $O(\log n\log\log n)$ bits.
It can be computed in time $O(\log n)$.
\item Sort the blocks and store the resulting BWT, $gt$ and sampled inverse
suffix arrays in external memory. Using the suffix and LCP arrays of the basic
blocks also compute the start ranks necessary for the backward searches when
merging the blocks together. This takes time $O(n\log n\log\log n)$ in the
worst case and $O(n)$ on average and space $O(b\log b) = O(\frac{n}{\log n}\log\frac{n}{\log
n})\allowbreak =O(n)$ 
bits of internal memory.
\item Process the merge tree. Mark all leafs as finished and all inner
nodes as unfinished. While there are unfinished nodes choose any unfinished
node with only finished children, merge the respective blocks and mark the
node as finished. There are $O(\log n)$ leafs and the tree is balanced, so
it has $O(\log\log n)$ levels. Each single level can be processed in time
$O(n\log\log n)$. So the total run time for the tree merging phase is
$O(n\log^2\log n)$. The maximum internal memory space usage appears when
performing the merge operation at the root of the tree. Here we need space
$b_l H_k + o(b_l \log\sigma)$ bits where $b_l$ denotes the sum of the length
of the blocks in the left subtree which is $O(n)$ and $H_k$ denotes the
$k$'th order entropy of the text comprising those text blocks.
\end{enumerate}
Summing over all steps the run-time of the algorithm is $O(n\log n\log\log n)$
in the worst case and $O(n\log^2\log n)$ on average. In practice this means
we can compute the BWT of a text as long as we are able to hold the text 
(more precisely the text for the left subtree of the merge tree) in internal memory.
If we can hold a fixed fraction of the text in main memory, then we can
still compute the BWT of the text in the same run-time by resorting to the
original iterative merging scheme from \cite{DBLP:journals/algorithmica/FerraginaGM12}.
We decompose the text into blocks of size $b'$ such that $b'\leq \frac{n\log\sigma}{c\log n}$
where $\frac{1}{c}$ is the fixed fraction of the text we can hold in
internal memory and compute a partial BWT for each of these blocks where the
suffixes sorted are considered as coming from the whole text, i.e.~suffix
comparisons are still over $\tilde{t}$ and not limited to a single of the blocks.
Then we merge these blocks along a totally skewed merge tree such that the
left block always has size about $b'$. The size of the set of partial sparse gap
arrays required at any time remains bounded by $O(n)$ bits. As the number of
blocks is fixed, the total asymptotical run-time of the algorithm remains
$O(n\log n\log\log n)$ in the worst case and $O(n\log^2\log n)$ on average.
\section{Parallelisation}
\label{sect:parallelisation}
Many of the steps of our algorithm are parallelisable. As building blocks
we will need representations of the required data structures in external memory
which allow accessing parts without performing a complete sequential scan
and we will require parallel versions of
\begin{itemize}
\item the construction of rank indexes,
\item the computation of gap arrays,
\item the merging of BWT blocks,
\item the merging of sampled inverse suffix arrays and
\item the sorting of single blocks.
\end{itemize}
As the merging of sampled (inverse) suffix arrays is basically a simplified
version of the merging of BWT blocks (the difference is not all values are
present), we will only describe the merging of BWT blocks and leave the
sub sampled merging variant as an exercise for the reader.
In this paper we consider shared memory parallelisation only, i.e.~a setting
where all processors involved have access to the same internal memory. Some
of the algorithms may be modified for distributed memory settings.
\subsection{Data Structures in External Memory}
We are using the following data structures in external memory:
\begin{itemize}
\item The $gt$ bit vectors
\item Burrows Wheeler transforms of blocks
\item Dense $\gamma$ coded gap arrays
\item Sparse $\gamma$ coded gap arrays
\end{itemize}
Arbitrary positions in the $gt$ bit vectors can be accessed without further
information, so no additional information is necessary to support scanning
starting from a given position.
For the Burrows Wheeler transform sequences we choose the following
representation. For a given string $B=b_0\ldots b_{m-1}$ over $\Sigma$
and a given block size $d\in O(\log^2 n)$ we partition $B$ into blocks
of size $d$ such that all but the last block have size $d$. Each block
factor is stored using run-length, Huffman and $\gamma$ code. We first compute a
run length representation of the block which gives us a sequence of pairs
comprised of symbols and number of consecutive occurrences. As an example the
factor $aabbbcaaa$ would be transformed to $(a,2),(b,3),(c,1),(a,3)$. This
sequence of pairs is stored using Huffman code for the symbols and Elias
$\gamma$ code for the run lengths. As we assume a constant alphabet size,
the dictionary for the Huffman code takes constant space per block. The
space used by the Huffman coded sequence of symbols is bounded by $d \lceil\log\sigma\rceil$.
According to our considerations about $\gamma$ code above the maximum amount
of space required for storing any sequence of run lengths for a block of
length $d$ is $d(1+2\lceil\log(1+1)\rceil)=3d$. So in the worst case the
representation of the block BWT in this code is in $O(d\log\sigma)$ bits.
In consequence each block can be addressed using a fixed size pointer of
length in $O(\log (m\log\sigma))$. As we have $d\in O(\log^2 n)$
and are only storing BWTs of length $m\leq n$ we are storing at most
$O(\frac{m}{\log^2 n})$ pointers. In total these require space $O(\frac{m}{\log n})$ bits.
When the BWT is coded in this way we can start decoding it at any index
which is a multiple of the block size $d\in O(\log^2 n)$. The setup procedure
takes $O(1)$ time and we can decode each following symbol in time $O(1)$
when performing a sequential scan 
(see \cite{moffat1997implementation,DBLP:conf/dcc/KarkkainenT13} for constant time
encoding and decoding of Huffman/minimum redundancy
code. The code lengths of the $\gamma$ code can be decoded using a lookup table taking $o(n)$
bits).
Dense $\gamma$ coded gap arrays can be stored very similar to our BWT
storage scheme for allowing the start of a decoding run from any position in
the array with very little overhead. As shown above an array $G$ of length
$l$ such that $\sum_{i=0}^{l-1} G[i]=s$ can be stored using $O(s+l)$ bits,
which means the start index for the code of any position can be stored in
$O(\log (s+l))$ bits. We decompose the array into non-overlapping blocks of
length $e\in O(\log^2 n)$ elements and store a pointer to the start of each
block. These pointers take space $O(\frac{s+l}{\log n})$ bits in total and we
can start the decoding from any index on $G$ after a setup time in $O(\log^2 n)$.
Subsequently each following element can be sequentially decoded in constant time per element.
For the sparse $\gamma$ coded gap files we have already described the code
itself above. 
Let us assume we are considering a gap array $G$ of length $\ell$
such that the sum over the elements of $G$ is $s$ and $G$ contains $k$ non-zero values.
To facilitate the start of a decoding run at a given index on $G$ we keep
the $\gamma$ code as is, where we assume that the two $\gamma$ coded
sequences (distance between non-zero values in the array and non-zero values) are stored interleaved. 
As the array is sparse we possibly would be adding excessive amounts of space
by adding pointers marking the next non-zero value for equidistant indices
on $G$. 
Instead we store information for every $j=\lceil\log n\rceil^2$'th instance of a non-zero
value on $G$ for $i=0,1,\lfloor\frac{k}{j}\rfloor$. 
For each such instance of a non-zero value we store the bit position of the code pair in the
$\gamma$ code and the absolute index on $G$.
The $\gamma$ code itself requires $O(k(1+\log\frac{s\ell}{k^2}))$ bits as shown above, so each pointer into the code takes $O(\log (k \log n))$
bits and there are $O(\frac{k}{\log^2 n})$ such pointers,
so the space used for pointers is $O\left(\frac{k \log(k\log n)}{\log^2 n}\right)$ bits
which is bounded by the space used for the sparse array itself. 
Storing the selected absolute values of non-zero indices on $G$ takes space $O(\frac{k}{\log n})$.
Using this index on the external memory representation for $G$ we can start
decoding the array from any index on $i$ on $G$ by first employing a binary
search on the absolute value of the largest indexed non-zero value element
in $G$ with an index not greater than $i$. This element may be $O(\log^2 n)$
non-zero elements away from our desired target which need to be skipped
before we reach the relevant portion of $G$. Thus the setup time is
$O(\log^2 n)$. After the setup we can decode either the sequence of non-zero
values starting from $i$ or the sequence of all values starting from $i$
in constant time for each such value. Decoding non-zero values is required
for merging two sparse arrays while decoding all values is required for
merging two BWTs.
For parallel encoding of these data structures we note that all of them are
easily writable in any given number of parts to a set of files such that a
simple combination of these files represents the respective complete structure.
Thus encoding the data structures does not need be performed sequentially.
If we write data to multiple files, then we need a simple additional layer
of indexing meta data designating the length of the corresponding sequence
in each file. For $gt$ bit vector, Burrows Wheeler transform and dense
$\gamma$ coded gap files this additional data stores the number of bits,
symbols and gap array entries respectively in the form of a prefix sum
array. This array allows us to find the correct starting file and offset in
this file given the offset in the complete sequence in logarithmic time in
the number of files. The number of files for any one such sequence is bounded by
the number of processors, so the space for the prefix sum array is not a
concern as long as the number of processors $p$ is $O(\frac{n}{\log n})$.
For sparse $\gamma$ coded gap files the situation is slightly more involved,
as we want to be able to start processing from a given index $i$ on the gap file
to process the entry at index $i$ and all gap file entries following $i$
or given some $j$ skip $j$ non-zero values from the front of the file and
process all following non-zero elements. For this purpose we store two
index sequences in addition. The first is a prefix sum sequence over the
number of gap array elements represented by each file. Here the number
stored in a file is obtained as
\begin{itemize}
\item the length of the gap array if there is only one file or
\item otherwise, if there is more than one file
\begin{itemize}
\item the index of the last non-zero gap array value represented in the file
plus one for the first file,
\item the index of the last non-zero gap array value represented in the
current file minus the index of the last non-zero gap array value
represented in the previous file if the file is neither the first nor the
last and
\item the length of the full gap array minus the index of the last non-zero
gap array value represented in the penultimate file for the last file.
\end{itemize}
\end{itemize}
The second is the prefix sum sequence over the sequence of non-zero values
stored in each file. Both additional index sequences store $O(p)$ values.
Again the space for these can be neglected when $p\in O(\frac{n}{\log n})$.
Searching the file level indexes for the correct file and in file offset for
a given global offset in each case takes time $O(\log p)$ and is thus not
critical in comparison with the in file access times for BWT and gap files. We
will thus without loss of generality below discuss a situation were the 
respective streams are assumed to be given as a single file instead of a 
set of files for the sake of simplicity of exposition.
\subsection{Parallel Construction of Rank Indexes}
There are several data structures which allow constant time rank
queries on sequences over a finite alphabet. In this paper we will consider
the wavelet tree (see \cite{DBLP:conf/soda/GrossiGV03}). 
A prefix free code $g$ for an alphabet $\Gamma$ is a function
$g:\Gamma\mapsto\{0,1\}^\ast$ such that for all pairs of symbols
$(a,b)\in\Gamma^2$ such that $a\ne b$ the code $g(a)$ is not a prefix of $g(b)$.
In the following we in addition assume that each prefix free code considered is such
that each binary string $c$ either has a prefix which is a code of some
symbol or $c$ is a prefix of at least one code produced by the code. 
This condition implies that there are no unused binary code words. 
In consequence the longest code assigned to any symbol is bounded by the
size of the input alphabet.
For any prefix free code $g:\Gamma\mapsto\{0,1\}^\ast$ where $|\Gamma|>1$
let $g_i:\Gamma_i\mapsto\{0,1\}^\ast$ denote the sub
code obtained from $g$ by choosing $\Gamma_i=\{a\in\Gamma \mid g(a)[0] = i \}$
and $g_i(a)=g(a)[1,|g(a)-1|]$. Informally $g_i$ is defined for all symbols
in $\Gamma$ for which the first bit of the code assigned by $g$ is $i$ and
this first bit is stripped off when transforming $g$ to $g_i$.
Let $f_i:\Gamma^\ast\mapsto\Gamma_i^\ast$ be the function given by the
homomorphism defined by mapping each symbol in $\Gamma_i$ to itself and all
other symbols to the empty word, i.e.~the function removing all symbols
which are not in $\Gamma_i$.
Let $h:\Gamma\mapsto \{0,1\}^\ast$ denote a prefix free code.
Then the wavelet tree concerning $h$ for a text $s\in\Gamma^\ast$ is obtained in
the following recursive way.
\begin{itemize}
\item If $|\Gamma|=1$ then produce a leaf.
\item Otherwise store the bit sequence $h(s_0)[0]h(s_1)[0]\ldots h(s_{|s|-1})[0]$
in an inner node. Attach a left and right child to this inner node where the
left child is the wavelet tree for $f_0(s)$ concerning the code $h_0$ and
the right child is the wavelet tree for $f_1(s)$ concerning the code $h_1$.
\end{itemize}
When this procedure is complete, then we assign a binary code to each inner
node, where we assign $\epsilon$ to the root, $c0$ to each inner node which
is a left child of a parent with code $c$ and $c1$ to each inner node which
is a right child of a parent with code $c$.
For the parallel construction of a wavelet tree based on a prefix free code
$h$ from an input string observe
that each interval of indices on the input string corresponds to a unique
interval of indices in the bit vector stored in any inner node of the
wavelet tree. For constructing a wavelet tree in parallel from an input
string $s$ of length $m$ over $\Sigma$ using $p$ processors we use the following steps.
\begin{enumerate}
\item Choose a block size $d=\lfloor\frac{m+p-1}{p}\rfloor$.
\item Assign index block $I_i=[di,\min(di+d,m))$ to processor $i$.
\item Each processor computes a symbol histogram $H_i:\Sigma\mapsto\mathbb{N}$
for its interval such that $H_i[a]=|\{j\mid j\in I_i\textnormal{ and } s[j]=a\}|$.
\item Compute the prefix sum arrays $\hat{H}_i$ defined by $\hat{H}_i[a] = \sum_{j=0}^{i-1} H_i[a]$ for all $a\in\Sigma$. This can be parallelised along the alphabet $\Sigma$.
\item Let $C(h)=\{h(a) \mid a\in\Sigma\}$ and let $\hat{C}(h) = \{ c' \mid c' = c[0,i]
\textnormal{ for some } c\in C(h) \textnormal{ and } -1\leq i < |c|-1 \}$ (the set of all proper prefixes of codes produced by $h$).
Let $P_i:\hat{C}\mapsto \mathbb{N}$ be defined by $P_i(j) = \sum_{k \mid j \textnormal{ prefix of h(k)}} \hat{H}(k)$.
Processor $i$ computes table $P_i$. Note that $P$ may be understood as a
function from integers to integers instead of binary strings to integers, as
the elements of $\hat{C}$ may be interpreted as the numerical value
represented by the respective binary strings. $P_i$ denotes the bit offset
of the input interval for input block $i$ for each inner node of the
constructed wavelet tree.
\item Processor $i$ produces the wavelet tree for interval $I_i$ while
instead of starting at the beginning of the bit vector of a node starting at
the bit position designated by the array $P_i$.
\item For each node of the wavelet tree the $p$ processors compute the
necessary prefix sums for the binary rank dictionaries of each node in parallel.
\end{enumerate}
This approach can be used to obtain a rank dictionary for the input string
$s$ of length $m$ over $\Sigma$ in time $O\left(\frac{m}{p}\right)$ using $p$ processors.
This applies for block type as well as Huffman code.
\subsection{Parallel Computation of Gap Arrays}
When we compute a gap array we are generally merging two adjacent text blocks.
Let the left and right block sizes be $b_l$ and $b_r$ respectively and let
the indices of the left block on $\tilde{t}$ start at index $i$. In the
serial version we start with the suffix at index $i+b_l+b_r-1$ and perform
$b_r-1$ steps of backward search using the index for the left block. This
can be parallelised by performing several backward searches in parallel
starting from several indices within the right block. If $p\leq b_r$ processors are
available, then we can choose the starting positions
$x_j = i+b_l+b_r-1-j\lceil\frac{b_r}{p}\rceil$ for $j=0,1,\ldots,p-1$ and
perform $\ell_j=x_{j+1}-x_j$ steps for $j=0,1,p-2$ and $\ell_j=b_r-x_j$
steps for $j=p-1$. As the backward search is based an suffix ranks in the
left block we need to find the ranks of these $p$ suffixes in the left
block using forward search on one or more suffix arrays just as we did above 
for the single suffix at the end of the right block in serial mode. 
As we have $p$ processors and need to find $p$ ranks, the run-time
stays the same as in the serial version for both average and worst case.
In the average case the forward searches are not a time critical factor as long as 
$p$ is $O(\frac{n}{\log^2 n\log\log n})$, so 
in practice we can expect the forward searches to be of no concern for the
execution time. When we perform $p$ backward searches in parallel, then we
also need to be prepared to perform the corresponding increment operations
on the gap array in parallel. 
This is no problem if the gap array can be
kept in internal memory as modern processors have specialised
operations for atomic, lock free modifications of single memory cells. If we
use sparse gap arrays in external memory, then the situation is somewhat
more complicated. In the serial version we first accumulate a number of
indices in the gap array whose value is to be incremented in a list of
length $k\in O(\frac{b_r}{\log^2 b_r})$ in internal memory.
If this list runs full, then it needs to be sorted and written to external memory as a sparse
$\gamma$ coded gap array. Whenever there are two sparse gap arrays of the
same sum, then the two need to be merged to obtain a single sparse array of
twice that sum (as the number of non-zero elements in the files approaches
the size of the left block we switch to dense arrays as described above.
Generating and merging dense $\gamma$ coded gap arrays are however
simplified versions of the equivalent problems in sparse arrays, so we
will not explicitly discuss this below as the algorithms are easily deduced
from the sparse case).
In the parallel case we keep the concept of a single array
accumulating indices in internal memory. All $p$ processors are appending
values to this array until it runs full. This can be implemented by using a
fill pointer which is modified by atomic increments while elements are
inserted at the array index designated by the fill pointer before the increment.
A process stalls backward search when the pointer exceeds a given threshold
signalling that the array is full. When all processes have detected that the
array is full, then we can perform a parallel radix sort on the array. As in
the serial case this is performed in two stages, i.e. using $e\in O(\sqrt{n})$
buckets. Like for histogram computation in the parallel wavelet tree
construction we partition the array into $p$ non overlapping intervals and
compute a separate histogram for each interval. This takes space
$O(p\sqrt{n}\log n)$ bits which is negligible as long as $p\in o(\frac{\sqrt{n}}{\log n})$
which is not an issue in practice. 
Prefix sums are computed based on these bucket histograms
to obtain the starting index for the elements for each
bucket originating from each processor input interval. The elements for
bucket $i$ appear ahead of those for buckets with index larger than $i$ and
within each bucket the elements originating from processor interval $j$
appear before those for intervals with higher indices. These prefix sums can
be computed in parallel by first accumulating all element counts over
$\lceil\frac{e}{p}\rceil$ buckets and then based on these offsets filling in
the final prefix sums in parallel. Using the computed prefix sums the
numbers can be sorted into their buckets in parallel while keeping the
stable order necessary for a correct radix sort. Using this radix sort
method we can sort a set of the $k$ numbers from $[0,b_l]$ in time $O\left(\frac{\max\{\sqrt{b_l},k\}}{p}\right)$
using $p$ processors. After the array is sorted we again partition it into
$p$ index intervals. In each interval in parallel we transform the now
sorted sequence into run length encoding, i.e.~pairs of numbers and the
number of occurrences of the number in the interval. This representation does
not asymptotically need more space than the sorted number sequence. The
computation of the run length encoding in intervals produces a correct run
length encoding of the complete sequence except for the borders of the
intervals. Here the last run length in one interval may be for the same
number as the first run length in the next interval. This however can be
corrected in time $O(\log p)$ using $p$ processors, where we assign the
complete run-length to the leftmost occurrence of the number and mark the
other occurrences as invalid. After this we can produce a compact version
without space between the intervals of the run length encoded array in 
internal memory in time $O(\frac{k}{p})$ by counting the number of runs per 
interval, computing the prefix sums over this number sequence and copying
the sequence to another array in compact form. Finally we decompose this
compact run length array into at most $p$ intervals of length
$\lceil\frac{\lceil\frac{k}{p}\rceil}{e}\rceil e$ where $e\in O(\log^2 n)$
is a positive integer and write the resulting decomposition to $p$ sparse $\gamma$ coded gap 
arrays such that the concatenation of the files represents the complete
sparse gap array. As each interval size but the last is a multiple of
$e$ each file starts on an index of the sequence of run
length pairs which is also a multiple of $e$. This file
production also takes time $O(\frac{k}{p})$. Finally we can also produce the
block pointers required for accessing the file without performing a linear
scan in the same time complexity. Summarising we can transform the
array of length $k$ accumulating indices for incrementing in the gap array to a sparse
$\gamma$ coded representation in external memory in time $O(\frac{k+\sqrt{b_r}}{p})$ using $p$ processors.
Now assume we have two sparse $\gamma$ coded gap arrays with $k_a$ and $k_b$
non-zero values respectively which are to be merged. We want to perform
this merging using $p$ processors, where in the optimal case each processor
is assigned the same amount of work. 
Both files store $e\in O(\log^2 n)$ non-zero values per block. 
Using a binary search on these values we can determine the smallest $i$ such that at least $\lceil\frac{k_a+k_b}{p}\rceil$
blocks are stored in the two files for values smaller than $i$.
The first processor merges the values of the two input sequences from the beginning up to value $i$.
The same scheme is used to split the remaining data into a work package for the second processor and the rest for the remaining processors etc. 
The splitting operations take time $O( p \log^3 n )$. 
We are computing $p$ splitting points, each such splitting point is computed
using a binary search on $O(n)$ sorted values and accessing each such value
takes time $O(\log^2 n)$.
After the splitting the merging is easily done in
parallel in time $O(\frac{k_a+k_b}{p})$ so the run time of the method is
$O(p \log^3 n + \frac{k_a+k_b}{p})$.
This is a convenient procedure when the
number of processors is low compared to the size of the input. If the number
of processors should not be neglectable in comparison to the input size,
i.e.~if $O(p\log^3 n)$ would be large, then we could employ a different
splitting method where we would decompose the merging work along an optimally
balanced binary tree and thus reduce the work for finding the splitting points
to $O(\log p \log^3 n)$. We would still be computing $O(p)$ splitting
points, but the splitting point search within each level of the tree of
depth $O(\log p)$ could be performed in parallel in time $O(\log^3 n)$.
\subsection{Merging Burrows Wheeler Transforms}
For merging two BWT sequences of length $b_l$ and $b_r$ respectively 
let us assume we have the two sequences in external memory and the gap 
array in internal or external memory. For parallel merging of the two
sequences we first compute the sum over the gap array for each interval 
of integer length $e\in O(\log^2 n)$ starting at index $ie$ on the gap array $G$
for $i=0,1,\ldots,\frac{b_l}{e}$ using $p$ processors in parallel in time
$O(\frac{b_l}{p})$. These can be produced sequentially and in external
memory and take space $o(b_l)$ bits. Then we compute the prefix sums over
this sequence in parallel in time $O(\frac{b_l}{p})$ and also store it in
external memory. Using these prefix sums we can assign approximately equal
work loads to each processor. Using a binary search on the prefix sums we
can search for the smallest $j<b_l+1$ such that $s = \sum_{i=0}^{j} (G[i]+1) \geq \lceil\frac{b_l+b_r}{p}\rceil$
and then increase $j$ and update the sum $s$ accordingly until 
$\lceil\frac{s}{e}\rceil e - s \leq G[j+1]$ or $j=b_l+1$ if no such $j$ exists (remember that $G$ has length $b_l+1$).
The binary search gives us a work package which is approximately $\frac{1}{p}$
of the total length. The subsequent increasing of $j$ is for allowing us to produce
a file with a length which is a multiple of the block size by just
considering the next gap array value $G[j+1]$. The case where there is no
such $j$ can only occur for very sparse arrays with high run length values.
The first processor merges the front parts of the two BWTs by in turn taking
$G[i]$ symbols from the second BWT stream and one symbol from the first one
for $i=0,1,\ldots,j$ and finally $\lceil\frac{s}{e}\rceil e - s$ from the
second BWT stream to obtain a multiple of the block size $e$. The rest of
the two BWT streams is merged by the other processors. The further splitting
of the remaining data into blocks is performed just like the splitting of
the data into the part for the first processor and the rest. The only
difference consists in the fact that a fraction of $G[j+1]$ is already
handled by the first processor and thus needs to be excluded from the
remaining process. As for merging sparse $\gamma$ coded gap arrays the
finding of the splitting points in this manner takes time $O(p \log^3 n)$
following a similar argumentation. The actual merging is then performed in time
$O(\frac{b_l+b_r}{p})$ using $p$ processors. In practice it may be desirable
to perform the merging directly using the run length encoding instead of
symbol per symbol. In the asymptotical worst case there is however no
difference in the run time. Like in the case of merging sparse $\gamma$
coded gap files we can speed up the computation of the splitting points to
$O(\log p\log^3 n)$ if necessary. The approach of parallelisation of BWT
merging can be extended to multi way instead of two way merging as described
in \cite{paper9:ICABD2014}. The concepts remain the same, but the searching
of split points is over all of the involved prefix sum arrays of gap files instead of
just one. The merging within the single merge packets remains the same as in
the serial version.
\subsection{Parallel Sorting of Single Blocks}
We need to suffix sort single blocks and produce the respective block BWTs
before we can merge these structures to obtain the final BWT of the complete
text. To our best knowledge there is currently no parallel in place suffix
sorting algorithm. The external memory algorithmic variants described in
\cite{Karkkainen:2006:LWS:1217856.1217858} either distribute a total work of $\Omega(n\log n)$
over $p$ processors or work in some randomised model of computation.
As the problem of parallel suffix sorting in place is still open, we
resort to using our parallel block merging scheme to achieve some
parallelism for the average case. In the highly unlikely worst case this will again not
bring any improvement over the serial version. The block size $b$ in the
serial version is chosen to allow the suffix sorting of a single block in
internal memory. For the parallel case we deduce a reduced block size $b_p = \lceil\frac{b}{p}\rceil$
which allows us to perform the suffix sorting of $p$ such blocks at the same
time in internal memory. After these blocks have been handled we merge
the resulting BWTs until we obtain the BWTs for block size $b$ using our
parallel merging approach described above. Note that this does not require
us to store gap arrays in external memory, as we have sufficient space in
internal memory to perform the merging. In the unlikely worst case we will
spend time $\Omega(n\log n\log p)$ for finding the start ranks of our
backward searches during gap array construction. On average the run time is
$O(\frac{b\log p}{p})$ in all practically relevant cases (i.e.~when $p\log p\in O(\frac{n}{\log^4 n})$).
\subsection{Parallelisation of the Complete Algorithm}
Using the building blocks described above we can now present the parallel
version of our algorithm for $p$ processors. The steps are:
\begin{enumerate}
\item Choose a target block size $b''\in O(\frac{n}{\log n})$ as in the
serial version. Transform this block size $b''$ to
$b'=\lceil\frac{b''}{p}\rceil$, a target block size for parallel computing
and finally deduce a block size $b=\lceil\frac{n}{\lceil\frac{n}{b'}\rceil}\rceil$.
Split the text $t$ into blocks of size $b$ and $b-1$ only such that all the
blocks of length $b$ appear before the ones for $b-1$. This step takes
constant time and space.
\item Compute which blocks in $\tilde{t}$ propagate repetitions of a period at most
$b$ and for each block which is followed by a block propagating a repetition
whether it is generating this repetition. We can handle $p$ blocks in
parallel during the propagation as well as the generation checking phase as
each block can be handled independently and we have sufficient space in
internal memory due to the reduced block size compared to the serial
version. This step consequently takes time $O(\frac{n}{p})$ in total and
space $O(p\frac{n\log\sigma}{\log n p})=O(\frac{n\log\sigma}{\log n})\subseteq o(n)$ bits. The result can be stored in external memory.
\item Compute a balanced merge tree for the blocks in serial like for the
serial cases. This takes time $O(\log n)$ and space $o(n)$ bits.
\item Sort the blocks and store the resulting BWT, $gt$ and sampled inverse
suffix arrays in external memory while processing $p$ blocks at a time.
Using the suffix and LCP arrays of the basic blocks also compute the start
ranks necessary for the backward searches when merging blocks together. This
step takes time $O(n\log n (\log\log n+\log p))$ in the worst case and $O(\frac{n}{p})$
on average and space $O(b'\log b')=O(n)$ bits of internal memory.
\item Process the merge tree. Mark all leafs as finished and all inner nodes
as unfinished. While there are unfinished nodes choose any unfinished node
with only finished children, merge the respective blocks and mark the node
as finished. When merging blocks use an internal memory gap array if the node has a
distance of at most $\lceil\log p\rceil$ from any leaf in the underlying sub tree
and sparse $\gamma$ coded gap arrays in external memory otherwise. Use the
parallel versions of gap array computation, merging of $\gamma$
coded gap array files, merging of BWT streams and merging sampled inverse
suffix array files. The $\log\log n$ levels closest to the root can be
processed in total time $O(\frac{n\log\log n}{p})$ per level or $O(\frac{n\log^2\log
n}{p})$ in total. The $\log p$ other (lowest) levels of the tree can be
processed in time $O(\frac{n\log p}{p})$. The total run-time for the merging
stage is thus $O(\frac{n}{p}\max\{\log^2 \log n, \log p\})$. The space
requirement in internal memory remains asymptotically the same as in the serial
version as long as $p\in o\left(\frac{\sqrt{n}}{\log n}\right)$ as described above.
\end{enumerate}
The total run-time of the parallel algorithm using $p$ processors is thus
$O(\frac{n}{p}\max\{\log^2 \log n, \log p\})$ on average and with high
probability and $O(n\log n (\log\log n+\log p))$ in the worst case. As in
the serial case we can resort to the original skewed merging approach on top
of the balanced merging if there is not sufficient internal memory to hold the index for the left half 
of the text in memory but sufficient to keep an index for a fixed fraction
of the text while keeping the same run-time and space bounds.
\section{Conclusion}
\label{conclusion}
We have presented a new semi external algorithm for computing the Burrows-Wheeler transform
designed for the case where we can keep a fixed fraction of the input text in internal memory.
On average our new algorithm runs in time $O(n\log^2\log n)$ and is faster then the
algorithm of Ferragina et al published in \cite{DBLP:journals/algorithmica/FerraginaGM12}
while in the worst case it is only slower by a factor of $O(\log\log n)$.
In comparison with the algorithm by Beller et in \cite{DBLP:conf/spire/BellerZGO13}
our algorithm can be applied for the case when less than $8$ bits per symbol
of internal memory are available. We have also presented a parallel version
of our algorithm which on average and with high probability runs in time
$O(\frac{n}{p}(\log^2\log n + \log p))$ on $p$ processors and should achieve
this run time in all practically relevant applications. In the worst case the run
time of the parallel algorithm is $O(n\log n (\log\log n+\log p))$ and thus,
depending on the number of processors $p$, possibly even slower than the
serial version. However this worst case is highly unlikely in practice.
The practically interested reader can find a partial implementation of the
algorithms presented in this paper at https://github.com/gt1/bwtb3m.
The current version of the code does not implement the repetition reduction
as described in Section \ref{singleblocksorting} but uses a simplified
approach computing the length of the longest common prefix of all suffixes
in a block with the first suffix of the following block for determining how
far a block needs to be extended to obtain a correct sorting order. On
average this extension is by $O(\log n)$ characters. For texts featuring
long repetitions of short periods (shorter than the block size) however the method 
would be problematic.
One major open problem in our approach is the worst case run time. While it
is not as extreme as the worst case of some other algorithms for suffix sorting or BWT
construction showing good average case behaviour (cf.~\cite{Puglisi:2007:TSA:1242471.1242472}), 
it still leaves room for improvement. This is particularly true for the case of the parallel variant of the algorithm.
The increased worst case run-time over the average case stems solely from
the forward searches on block suffix arrays to obtain the starting points for the
backwards searches during gap array construction. In the serial version a
string needs to feature a very large LCP value in $\Omega(\frac{n\log\log n}{\log n})$
between two of its suffixes to trigger the worst case behaviour, in the parallel case using $p$
processors a LCP value in $\Omega(\frac{n\log\log n}{p \log n })$ is sufficient.
So an efficient method for detecting such high LCP values using small space
would be necessary to alleviate the effects of the problem. Note that these
high LCP values are not necessarily caused by runs (cf.~\cite{KolpakovKucherovFOCS99})
as the positions of the suffixes in the string sharing a long common prefix
may be too far apart to allow the common strings to connect, i.e.~build a square.
Two other open problems are the actual practical space usage of our algorithm in
external memory and its I/O complexity. 
While the space usage in external memory is asymptotically clear with 
$O(n\log\sigma)$ bits the involved constants are important in practice. Punctual observations in our
implementation studies suggest that the involved constants are low, but a
systematic analysis is desirable. A very simple analysis in analogy to the
space usage shows that the algorithm uses $O(n \log^2 \log n)$ bits of I/O
for merging $\gamma$ coded gap arrays, $O(n \log \sigma \log \log n)$
bits of I/O for merging BWTs and $O(n \log\log n)$ bits of I/O for writing
and reading $gt$ arrays, where I/O is in all cases streaming. As for the
space usage in external memory, it would also be interesting to determine
the involved constants in future work.

\bibliographystyle{abbrv}
\tiny
\bibliography{main}

\begin{thebibliography}{10}

\bibitem{DBLP:conf/spire/BellerZGO13}
T.~Beller, M.~Zwerger, S.~Gog, and E.~Ohlebusch.
\newblock {S}pace-{E}fficient {C}onstruction of the {B}urrows-{W}heeler
  {T}ransform.
\newblock In O.~Kurland, M.~Lewenstein, and E.~Porat, editors, {\em SPIRE},
  volume 8214 of {\em Lecture Notes in Computer Science}, pages 5--16.
  Springer, 2013.

\bibitem{DBLP:conf/alenex/BingmannFO13}
T.~Bingmann, J.~Fischer, and V.~Osipov.
\newblock {I}nducing {S}uffix and {LCP} {A}rrays in {E}xternal {M}emory.
\newblock In P.~Sanders and N.~Zeh, editors, {\em ALENEX}, pages 88--102. SIAM,
  2013.

\bibitem{burrows1994block}
M.~Burrows and D.~Wheeler.
\newblock {A} {B}lock-{S}orting {L}ossless {D}ata {C}ompression {A}lgorithm.
  {D}igital {S}ystems {R}esearch {C}enter.
\newblock {\em RR-124}, 1994.

\bibitem{DBLP:conf/cpm/CrochemoreGKL13}
M.~Crochemore, R.~Grossi, J.~K{\"a}rkk{\"a}inen, and G.~M. Landau.
\newblock {A} {C}onstant-{S}pace {C}omparison-{B}ased {A}lgorithm for
  {C}omputing the {B}urrows-{W}heeler {T}ransform.
\newblock In J.~Fischer and P.~Sanders, editors, {\em CPM}, volume 7922 of {\em
  Lecture Notes in Computer Science}, pages 74--82. Springer, 2013.

\bibitem{crochemore2007algorithms}
M.~Crochemore, C.~Hancart, and T.~Lecroq.
\newblock {\em {A}lgorithms on {S}trings}.
\newblock Cambridge University Press, 2007.

\bibitem{DBLP:journals/jea/DementievKMS08}
R.~Dementiev, J.~K{\"a}rkk{\"a}inen, J.~Mehnert, and P.~Sanders.
\newblock {B}etter external memory suffix array construction.
\newblock {\em ACM Journal of Experimental Algorithmics}, 12, 2008.

\bibitem{elias1975universal}
P.~Elias.
\newblock Universal codeword sets and representations of the integers.
\newblock {\em Information Theory, IEEE Transactions on}, 21(2):194--203, 1975.

\bibitem{DBLP:journals/algorithmica/FerraginaGM12}
P.~Ferragina, T.~Gagie, and G.~Manzini.
\newblock {L}ightweight {D}ata {I}ndexing and {C}ompression in {E}xternal
  {M}emory.
\newblock {\em Algorithmica}, 63(3):707--730, 2012.

\bibitem{ferragina2000opportunistic}
P.~Ferragina and G.~Manzini.
\newblock {O}pportunistic {D}ata {S}tructures with {A}pplications.
\newblock In {\em Foundations of Computer Science, 2000. Proceedings. 41st
  Annual Symposium on}, pages 390--398. IEEE, 2000.

\bibitem{DBLP:conf/soda/GrossiGV03}
R.~Grossi, A.~Gupta, and J.~S. Vitter.
\newblock High-order entropy-compressed text indexes.
\newblock In {\em SODA}, pages 841--850, 2003.

\bibitem{DBLP:conf/focs/HonSS03}
W.-K. Hon, K.~Sadakane, and W.-K. Sung.
\newblock {B}reaking a {T}ime-and-{S}pace {B}arrier in {C}onstructing
  {F}ull-{T}ext {I}ndices.
\newblock In {\em FOCS}, pages 251--260. IEEE Computer Society, 2003.

\bibitem{paper9:ICABD2014}
J.~K\"{a}rkk\"{a}inen and D.~Kempa.
\newblock Engineering a lightweight external memory suffix array construction
  algorithm.
\newblock In C.~S. Iliopoulos and A.~Langiu, editors, {\em
  2\textsuperscript{nd} {I}nternational {C}onference on {A}lgorithms for {B}ig
  {D}ata ({ICABD}2014)}, number 1146 in CEUR-WS Proceedings, pages 53--60,
  Aachen, 2014.

\bibitem{karkkainen2003simple}
J.~K{\"a}rkk{\"a}inen and P.~Sanders.
\newblock {S}imple {L}inear {W}ork {S}uffix {A}rray {C}onstruction.
\newblock In {\em Automata, Languages and Programming}, pages 943--955.
  Springer, 2003.

\bibitem{Karkkainen:2006:LWS:1217856.1217858}
J.~K\"{a}rkk\"{a}inen, P.~Sanders, and S.~Burkhardt.
\newblock Linear work suffix array construction.
\newblock {\em J. ACM}, 53(6):918--936, Nov. 2006.

\bibitem{DBLP:conf/dcc/KarkkainenT13}
J.~K{\"a}rkk{\"a}inen and G.~Tischler.
\newblock Near in place linear time minimum redundancy coding.
\newblock In A.~Bilgin, M.~W. Marcellin, J.~Serra-Sagrist{\`a}, and J.~A.
  Storer, editors, {\em DCC}, pages 411--420. IEEE, 2013.

\bibitem{DBLP:conf/cpm/KasaiLAAP01}
T.~Kasai, G.~Lee, H.~Arimura, S.~Arikawa, and K.~Park.
\newblock Linear-time longest-common-prefix computation in suffix arrays and
  its applications.
\newblock In A.~Amir and G.~M. Landau, editors, {\em CPM}, volume 2089 of {\em
  Lecture Notes in Computer Science}, pages 181--192. Springer, 2001.

\bibitem{knuth1977fast}
D.~E. Knuth, J.~H. Morris, Jr, and V.~R. Pratt.
\newblock {F}ast {P}attern {M}atching in {S}trings.
\newblock {\em SIAM Journal on Computing}, 6(2):323--350, 1977.

\bibitem{KolpakovKucherovFOCS99}
R.~Kolpakov and G.~Kucherov.
\newblock Finding maximal repetitions in a word in linear time.
\newblock In {\em Proceedings of the 1999 Symposium on Foundations of Computer
  Science (FOCS'99), New York (USA)}, pages 596--604, New-York, October 17-19
  1999. IEEE Computer Society.

\bibitem{manber1993suffix}
U.~Manber and G.~Myers.
\newblock {S}uffix {A}rrays: a {N}ew {M}ethod for {O}n-line {S}tring
  {S}earches.
\newblock {\em SIAM Journal on Computing}, 22(5):935--948, 1993.

\bibitem{moffat1997implementation}
A.~Moffat and A.~Turpin.
\newblock On the implementation of minimum redundancy prefix codes.
\newblock {\em IEEE Transactions on Communications}, 45(10):1200--1207, 1997.

\bibitem{navarro2007compressed}
G.~Navarro and V.~M{\"a}kinen.
\newblock {C}ompressed {F}ull-{T}ext {I}ndexes.
\newblock {\em ACM Computing Surveys (CSUR)}, 39(1):2, 2007.

\bibitem{nong2011two}
G.~Nong, S.~Zhang, and W.~H. Chan.
\newblock {T}wo {E}fficient {A}lgorithms for {L}inear {T}ime {S}uffix {A}rray
  {C}onstruction.
\newblock {\em Computers, IEEE Transactions on}, 60(10):1471--1484, 2011.

\bibitem{DBLP:conf/spire/OkanoharaS09}
D.~Okanohara and K.~Sadakane.
\newblock {A} {L}inear-{T}ime {B}urrows-{W}heeler {T}ransform {U}sing {I}nduced
  {S}orting.
\newblock In J.~Karlgren, J.~Tarhio, and H.~Hyyr{\"o}, editors, {\em SPIRE},
  volume 5721 of {\em Lecture Notes in Computer Science}, pages 90--101.
  Springer, 2009.

\bibitem{Puglisi:2007:TSA:1242471.1242472}
S.~J. Puglisi, W.~F. Smyth, and A.~H. Turpin.
\newblock A taxonomy of suffix array construction algorithms.
\newblock {\em ACM Comput. Surv.}, 39(2), July 2007.

\bibitem{szpankowski1991height}
W.~Szpankowski.
\newblock {O}n the {H}eight of {D}igital {T}rees and {R}elated {P}roblems.
\newblock {\em Algorithmica}, 6(1-6):256--277, 1991.

\end{thebibliography}

\normalsize
\newpage
\appendix

\end{document}